\documentclass[conference]{IEEEtran}
\IEEEoverridecommandlockouts

\usepackage{cite}
\usepackage{amsmath,amssymb,amsfonts}
\usepackage{algorithmic}
\usepackage{graphicx}
\usepackage{textcomp}
\usepackage{xcolor}
\def\BibTeX{{\rm B\kern-.05em{\sc i\kern-.025em b}\kern-.08em
    T\kern-.1667em\lower.7ex\hbox{E}\kern-.125emX}}
    
\usepackage{acronym}

\newacro{iid}[i.i.d.]{independent and identically distributed}
\newacro{RDC}[RDC]{risk decay coefficient}
\newacro{EER}[EER]{expected excessive risk}

\newacro{PSD}[PSD]{positive semi-definite}
\newacro{NSD}[NSD]{negative semi-definite}

\newacro{KL}[KL]{Kullback--Leibler}
\newacro{LLM}[LLM]{large language model}
\newacro{FIM}[FIM]{Fisher information matrix}
\newacro{ICL}[ICL]{in-context learning}
\newacro{MLE}[MLE]{maximum likelihood estimator} 

\usepackage{mathtools}

\providecommand \coleq {\vcentcolon =} 
 
\DeclareMathOperator{\oTranspose}{T}
\DeclareMathOperator{\oTr}{tr}
\DeclareMathOperator{\KL}{KL}
\DeclareMathOperator{\argmax}{argmax}
 \DeclareMathOperator {\oF} {F}
  \DeclareMathOperator {\diag} {diag}

\newtheorem{theorem}{Theorem}
\newtheorem{corollary}{Corollary}
\newtheorem{remark}{Remark}

\providecommand \T {^{\oTranspose}}
\providecommand \tr {\oTr}
\providecommand \KLD {D_{\KL}}
\providecommand \Jac {\nabla\T}

\providecommand \FIM {J}
\providecommand \eer {\ell}

\providecommand \nDem {n}
\providecommand \dimPar {K}
\providecommand \xAll {x^{\nDem}}
\providecommand \yAll {y^{\nDem}}
\providecommand \rxAll {X^{\nDem}}
\providecommand \ryAll {Y^{\nDem}}
\providecommand \xQ {x_{\mathrm Q}}
\providecommand \yQ {y_{\mathrm Q}}

\providecommand \vPar {{\alpha}}
\providecommand \vTPar {\underline{\vPar}}
\providecommand \domPar {\mathcal A}
\providecommand \vMLEPar {\underline{\vPar}^*}
\providecommand \pSingle [1] {P_{Y|X}^{(#1)}}
\providecommand \pMulti [1] {P_{Y^{\nDem} | X ^{\nDem} }^{(#1)}}

\providecommand \demFIM {\bar J}

\begin{document}

\title{A Theoretical Interpretation of In-Context Learning via Probabilistic Modeling}

\author{
		Zhenyu Liu, Huaze Tang, and Shao-Lun Huang
    \thanks{
    The authors are with Tsinghua Shenzhen International Graduate School, Tsinghua University, Shenzhen, China.
    }
    \thanks{
    The fundamental research described in this paper was supported, in part, by National Natural Science Foundation of China under Grant 62571297, and by Shenzhen Science and Technology Program under Grant JCYJ20240813112301003. 
    \textit{Corresponding author: Shao-Lun Huang (e-mail: {\texttt{twn2gold@gmail.com}})}
    }
}

\maketitle

\pagestyle{plain}
\thispagestyle{plain}
\pagenumbering{arabic}

\begin{abstract}
\Ac{ICL} is an emerging paradigm that employs the semantic information inherent in \acp{LLM} for generating answers to user queries. While the remarkable performance of \ac{ICL} has been widely known, a general modeling and a rigorous theoretical analysis of this paradigm are still lacking. This work presents a probabilistic model for \ac{ICL} and derives the performance of \ac{ICL} 
for both general parametric distributions and exponential families. Based on the derived results, the work explains the impact of multiple factors such as the number of demonstrations, the sensitivity of the probabilistic model to the variation of its parameters, as well as the similarity between the demonstrations and the query on the performance of \ac{ICL}. 
\end{abstract}


\acresetall		

\section{Introduction}
\Acp{LLM} has been successfully applied in different research fields including communication, sensing, and optimization \cite{XuNiyKan:24, ZhaGaoYu:25,LiuAstSee:24}. \Ac{ICL} is a promising paradigm for learning built upon the capability of \acp{LLM} \cite{DonLiDai:24}. In \ac{ICL}, an \ac{LLM} generates an answer to a user query based on a few demonstrations provided beforehand via the context prompt. Specifically, \ac{ICL} exploits the semantic knowledge inherent in \acp{LLM}  for responding to user queries without changing the pre-trained \ac{LLM} parameters. As a result, \ac{ICL} can achieve desirable performance using only a few demonstrations without  significant computational overhead.  

The remarkable performance of \ac{ICL} has piqued significant research interest in understanding its mechanism via theoretical analysis. In particular, recent studies have proposed multiple explanations for the mechanism of \ac{ICL}. 
One prominent line of research interprets \ac{ICL} through the lens of implicit meta-learning or Bayesian inference. In this view, the prompt implies a latent task, and transformers adapt by approximating Bayesian model averaging over the space of potential tasks \cite{xie2021explanation,dai2023can,li2023transformers,zhou2023algorithms}. Complementary studies emphasize an algorithmic perspective, demonstrating that transformers can simulate standard estimation algorithms. For instance, it has been shown that the attention mechanism can approximate gradient descent updates, with deeper layers effectively performing iterative optimization steps \cite{garg2022can,akyureklearning,von2023transformers,pan2023context}.

Another perspective highlights the critical role of pre-training data distribution. Research suggests that sufficiently broad task coverage and diversity during pre-training are essential for the emergence of robust in-context generalization \cite{raventos2023pretraining,chan2022data}. 
Conversely, narrow pre-training may lead to biased predictors that fail to generalize to new tasks. 
This connects to the practical challenge of prompt engineering, where performance is highly sensitive to the selection and ordering of demonstrations. While various retrieval-based and heuristic strategies have been proposed to select effective demonstrations \cite{wang2023large,su2022selective,qin2023context}, a rigorous theoretical understanding of why certain demonstrations are more effective than others remains limited.
%
To address the complexity of practical \acp{LLM} and bridge the gap between theory and practice, \cite{TanPenHua:26} proposed a linear probabilistic model for the conditional distribution of the answer given the input text.  
%
Based on this model, the work develops a theory to explain the effects of multiple factors 
on the \ac{ICL} performance.   This theory, while providing significant interpretation of \ac{ICL}, still has a gap with practice as the linear  probabilistic model  may not fit with existing \acp{LLM}. 

In this work, we propose a probabilistic model for \ac{ICL}, analyze the  performance of \ac{ICL} based on this model, and present interpretations of the theoretical analysis. In particular, an input text and its corresponding answer are both modeled as random variables. The \ac{LLM} uses the demonstrations to learn the conditional distribution of the answer given the input text. We adopt the \ac{EER}, namely the expected \ac{KL} divergence between the ground-truth conditional distribution and the one inferred by the \ac{LLM}, as the performance metric of \ac{ICL}. We derive the asymptotic \ac{EER} for both general parametric distributions and exponential families, and establish a non-asymptotic upper bound of \ac{EER} for exponential families. Building on these results, we explain the effect of the following factors on \ac{ICL}: the number of demonstrations, the sensitivity of the conditional distribution to the variation of parameters, and the similarity between the demonstrations and the query. Key contributions of this work are summarized in the following: 
\begin{itemize}
\item we propose a probabilistic model for \ac{ICL} and derive the asymptotic \ac{EER} for general parametric distributions;
\item we derive the asymptotic and non-asymptotic performance of \ac{EER} for exponential families; and 
\item we illustrate the effects of multiple factors 
on the performance of \ac{ICL}.
\end{itemize}

\textit{Notation}: The  $\ell_1$ norm, Euclidean norm, and the $k$-th entry of a vector $a$ are denoted by $\| a\|_1$, $\| a\|$, and $[a]_k$, respectively. The transpose, trace, Frobenius norm, and determinant of a matrix $A$ are denoted by $A\T$, $\tr\{A\}$, $\| A\|_{\oF}$, and $\det(A)$, respectively.  An identity matrix is denoted by $I$. For symmetric matrices $A$ and $B$, relationship $ A \succeq  B$ and $A \preceq  B$ represent that $A - B$ is \ac{PSD}  and \acl{NSD}, respectively. Operators $\nabla f$ and $\nabla^2 f$ represent the gradient and the Hessian, respectively, of a function $f$. 

\section{Problem Formulation}

The overall procedure of \ac{ICL} is described in the following. First, a set of $\nDem$ demonstrations $\{(x_i, y_i)\}_{i=1}^{\nDem}$ are provided to the prompt, where $x_i \in \mathcal X$ and $y_i \in \mathcal Y$ represent the {$i$-th} input text and the answer for this text, respectively (see Fig.~\ref{fig:procedure}). Then, a query input text $\xQ$ is provided to the prompt. Using the demonstrations, the \ac{LLM} predicts an answer $\yQ$ for the query $\xQ$. 

\begin{figure}
    \centering
    \includegraphics[width=\linewidth]{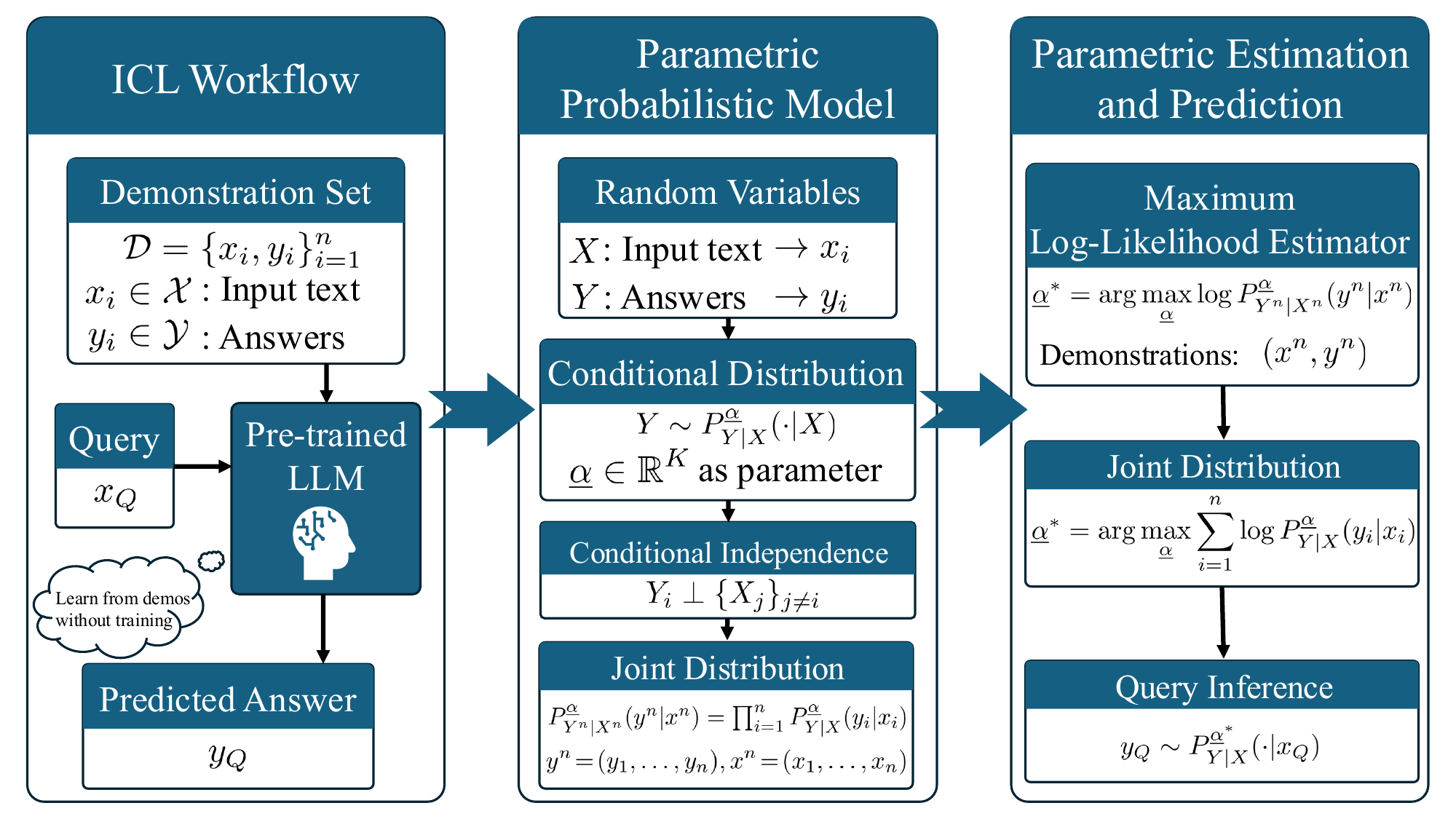}
    \caption{Schematic diagram of the in-context learning (ICL) framework and its underlying probabilistic modeling pipeline.}
    \label{fig:procedure}
\end{figure}

The following probabilistic model is adopted for 
\ac{ICL}. Both the input text $X$ and the answer $Y$ are modeled as random variables. As a result, $x_i$ and $y_i$ are realizations of random variables $X_i$ and $Y_i$, respectively.  Given an input text $X$, the answer $Y$ follows a parametric probabilistic model $\pSingle{\vTPar}$, where $\vTPar$ is a parameter consisting of $\dimPar$ entries and belongs to a set $\domPar \subseteq \mathbb R^{\dimPar}$.  Conditioned on the $i$-th input text $X_i$, the answer $Y_i$ is independent of $\{X_j\}_{j \neq i}$ and follows the conditional distribution $\pSingle{\vTPar}$. Under such conditional independence, the distribution of $\ryAll \coleq (Y_1, Y_2, \ldots Y_n)$ conditioned on $\rxAll \coleq (X_1, X_2, \ldots X_n)$ factorizes as
\begin{align*}
\pMulti{\vTPar} \bigl(\yAll \bigr | \xAll \bigr) = \prod_{i=1}^n \pSingle{\vTPar} (y_i | x_i)
\end{align*}
where $\yAll \coleq (y_1, y_2, \ldots, y_n )$ and $\xAll \coleq (x_1, x_2, \ldots, x_n)$. 
%
Using the demonstrations $(\xAll, \yAll)$, the \ac{LLM} computes an estimator of $\vTPar$. In particular, we consider the \ac{MLE} $\vMLEPar$ given by 
\begin{align*}
\vMLEPar \coleq   \argmax_{\vPar \in \mathcal A }  \, \log \pMulti{\vPar} \bigl(\yAll \bigr | \xAll \bigr) .
\end{align*}
Then, the \ac{LLM} generates the answer $\yQ$ by sampling from the inferred conditional distribution $\pSingle{\vMLEPar}(\cdot | \xQ)$. Note that the conditional distribution $\pSingle{\vPar}$ for any $\vPar \in \domPar$ is determined via pre-training and is not affected by the demonstrations. Indeed, the demonstrations only affects the estimate of the parameter. 

The quality of the answer to the query is affected by the discrepancy between the inferred conditional distribution $\pSingle{\vMLEPar}(\cdot | \xQ)$ and the ground-truth  $\pSingle{\vTPar}(\cdot | \xQ)$. Ideally, these two distributions are equivalent so that the answer to the query is sampled from the ground-truth conditional distribution. In practice, such equivalence may not be achieved and we employ the \ac{KL} divergence $\KLD \bigl (\pSingle{\vTPar}  \big \| \pSingle{\vMLEPar} \bigr)$ as a metric for the discrepancy between these two conditional distributions. In particular, such a divergence is defined as
\begin{align*}
\KLD \bigl (\pSingle{\vTPar}  \big \| \pSingle{\vMLEPar} \bigr)
\coleq \mathbb E_{\pSingle{\vTPar}} \biggl[ \log \frac{\pSingle{\vTPar} (Y | \xQ) }{\pSingle{\vMLEPar} (Y | \xQ) } \, \biggr | \, X = \xQ \biggr] .
\end{align*}
This quantity is a random variable as $\vMLEPar$ is a function of $\ryAll$. Define the \emph{\ac{EER}} $\eer(\xQ; \xAll )$ as the conditional expectation of the \ac{KL} divergence, i.e., 
\begin{align}\label{eq:def_eer}
\eer(\xQ; \xAll ) \coleq
\mathbb E_{\ryAll | \rxAll } \bigl[\KLD \bigl (\pSingle{\vTPar}  \bigr \| \pSingle{\vMLEPar} \bigr)  \bigr | \rxAll = \xAll \bigr] .
\end{align}
This paper investigates properties of the \ac{EER} and how it is affected by the demonstrations and the query. 

We make a few assumptions on $\log \pSingle{\vPar}$, where $\vPar \in \domPar$ is a general value of the parameter. To that end, define the \ac{FIM} $\FIM(\vPar; x)$ of the parameter $\vPar$ for demonstration $(X, Y)$ conditioned on $X=x$ as
\begin{align}\label{eq:def_fim}
\FIM(\vPar; x) \coleq 
 -\mathbb E_{\pSingle{\vPar}}\Bigl[\frac{\partial^2   }{\partial {\vPar} \, \partial {\vPar\T} }  \log \pSingle{\vPar}(Y | x) \Bigr | X = x \Bigr]  \,. 
\end{align}
In particular, $\frac{\partial }{\partial \vPar}\log \pSingle{\vPar}$ is called the score function. 
The assumptions on $\log \pSingle{\vPar}$ are listed in the following. 

\begin{enumerate}
\item [1.] \label{asm:regular}  
(regularity condition) $\log  \pSingle{\vPar}$ satisfies  
\begin{align}
 \mathbb E_{\pSingle{\vPar}}\Bigl[\frac{\partial  }{\partial {\vPar}} \log \pSingle{\vPar}(Y | x)  \Bigr | X = x \Bigr] = 0, \nonumber  \\
  \qquad \forall x \in \mathcal X ,~ \forall \vPar \in \domPar.
 \label{eq:regular}
\end{align} 


\item [2.] 
(finite \ac{FIM}) There exists a constant $c_1 > 0$ such that 
\begin{subequations} \label{eq:finite_fim}
\begin{align}
\FIM(\vPar; \xQ) & \preceq c_1 I \label{eq:finite_fim_quest}\\
 \FIM(\vPar; x_i) & \preceq c_1 I,  \quad \forall \vPar \in \domPar , ~ i=1,2,\dotsc, \nDem .
 \label{eq:finite_fim_dem}
\end{align}
\end{subequations}

\item [3.] 
%
(singular values for Hessian bounded away from zero) There exists a constant  $c_2 > 0$ such that 
\begin{align} 
 \sigma_{\min} \Bigl(\frac{1}{n} \sum_{i=1}^n \frac{\partial^2}{\partial {\vPar} \, \partial {\vPar\T} }  \log \pSingle{\vPar}(y_i | x_i)   \Bigr)  & \geq c_2 , ~  \forall \vPar \in \domPar  
\label{eq:positive_fim} 
\end{align}
where $\sigma_{\min}(\cdot)$ represents the minimum singular value of the argument.

\item [4.] (finite third-order derivative) There exists a constant $c_3$ such that
\begin{subequations}
\begin{align}
\Bigl| \frac{\partial^3  }{\partial [\vPar]_j \partial [\vPar]_k \partial [\vPar]_l} \log \pSingle{\vPar}(y | \xQ ) \Bigr | & \leq c_3 \\
\Bigl| \frac{\partial^3 }{\partial [\vPar]_j \partial [\vPar]_k \partial [\vPar]_l} \log \pSingle{\vPar}(y | x_i )  \Bigr | & \leq c_3,  ~ \forall y \in \mathcal Y    \nonumber \\
& \hspace{-5cm} \forall \vPar \in \domPar, ~ i=1,2, \ldots, \nDem, \quad j,k,l = 1,2,\ldots, \dimPar .
\label{eq:finite_third_order_dem}
\end{align}
\end{subequations}

\item [5.] (finite fourth moment of score function) There exists a constant $c_4$ such that 
\begin{align} \label{eq:finite_fourth_moment}
\mathbb E_{\pSingle{\vTPar}} \Bigl [ \Bigl\| \frac{\partial  }{\partial \underline{\alpha}} \log \pSingle{\vTPar}(Y_i | x_i) \Bigr\|^4  \Bigr | X = x_i \Bigr] \leq c_4 \nonumber  \\
 i=1,2,\dotsc, \nDem . 
\end{align}
\label{asm:finite_fourth}

\end{enumerate}
Specifically,~\eqref{eq:regular} is a mild assumption and is required for establishing information inequality using \ac{FIM} \cite{Kay:B93}.

In addition to general parametric distributions $\pSingle{\vTPar}$, we also consider the case where $\pSingle{\vTPar}$ belongs to an exponential family. In this case,  $\pSingle{\vTPar}$ can be written as
\begin{align}\label{eq:exp_family}
\pSingle{\vTPar}(y|x) = P_{Y}(y)\exp\left(\underline{\alpha}^\text{T}f(x_,y) - b(\underline{\alpha}, x)\right)
\end{align}
where $f(x,y) \in \mathbb R^{\dimPar}$ represents the semantic embedding of $(x,y)$ generated by the \ac{LLM}; $b(\underline{\alpha}, x) \in \mathbb R$ is a normalization term and satisfies $b(\underline{\alpha}, x) = \log \bigl( \sum_y  P_Y(y)\exp\left(\underline{\alpha}^\text{T}f(x,y)\right) \bigr)$. In particular, $f(\cdot, \cdot)$ is determined via pre-training and is not affected by the demonstrations.

\section{Asymptotic performance of \ac{EER}}
We derive the asymptotic performance of  $\eer(\xQ; \xAll )$ as the number $\nDem$ of demonstrations goes to infinity for general parametric distributions $\pSingle{\vPar}$ satisfying~\eqref{eq:regular} to~\eqref{eq:finite_fourth_moment}. To that end, define symmetric random matrices $A_n$ and $M_n$ as
\begin{align}
A_n &:= -\frac{1}{n} \sum_{i=1}^n  \frac{\partial^2}{\partial \underline{\alpha} \partial  \underline{\alpha}^{\mathrm T}} \log \pSingle{\vTPar}(Y_i | x_i)  \label{eq:def_A_n} \\
M_n & := \frac{1}{n} \Bigl(\sum_{i=1}^n \frac{\partial \log \pSingle{\vTPar}(Y_i | x_i) }{\partial \underline{\alpha}}   \Bigr) 
\Bigl(\sum_{i=1}^n \frac{\partial \log \pSingle{\vTPar}(Y_i | x_i)}{\partial \underline{\alpha}^{\mathrm T}}   \Bigr)  . 
\label{eq:def_M_n} 
\end{align}
We also introduce the short notation 
\begin{align*}
\mathbb E \bigl[  \cdot \bigr |   \xAll \bigr] \coleq \mathbb E_{\pMulti{\vTPar}} \bigl[\cdot \bigr |  \rxAll = \xAll \bigr] .
\end{align*}
The following theorem shows that $\eer(\xQ; \xAll )$ decreases to zero at a rate of $1 / \nDem$.
\begin{theorem}\label{prop:asymptotics_general}
Under Assumptions~\eqref{eq:regular}--\eqref{eq:finite_fourth_moment}, it holds that
\begin{align} \label{eq:asymptotic_eer}
\lim_{n \to \infty} n \eer(\xQ; x^n) - r(\xQ; x^n) = 0
\end{align}
where $r(\xQ; \xAll)$ is defined as 
\begin{align}\label{eq:coefficient_general}
%
r(\xQ; x^n) \coleq   \frac{1}{2}\mathrm{tr}\Bigl\{& \mathbb E \bigl[ A_n^{-1} M_n A_n^{-1} |  x^n\bigr]  \, J(\underline{\alpha} ; x_{\mathrm Q})  \Bigr\} .
\end{align} 
\end{theorem}
\begin{IEEEproof}
See~Appendix.
\end{IEEEproof}

\begin{remark}
Theorem~\ref{prop:asymptotics_general} shows that the asymptotic \ac{EER} for sufficiently large $\nDem$ is $r(\xQ; x^n)  / \nDem$. First, the asymptotic \ac{EER} is linear with respect to $r(\xQ; x^n)$, which is referred to as the \emph{\ac{EER} coefficient}. This coefficient is determined by both the demonstrations and the query. Second, the asymptotic \ac{EER} decays at a rate of $1/\nDem$. This is because that the \ac{LLM} can accurately infer the true parameter $\vTPar$ when $\nDem$ is sufficiently large, as indicated by~\eqref{eq:MLE_approaches_true} in the proof. Note that Theorem~\ref{prop:asymptotics_general} is different from asymptotic results for classical inference problems where the data are assumed \acl{iid}. Instead, it is only assumed that the conditional distributions of the answers in the demonstrations are identical. In practice, the demonstration inputs $\{X_i\}_{i=1}^{\nDem}$ are not necessarily independent. For example, if $\{X_i\}_{i=1}^{\nDem}$ are provided to the prompt sequentially with $i$ indicating the order, then $X_i$ can depend on $X_j$ for $j < i$. 
\end{remark}

Next, we show the asymptotic performance of the \ac{EER} for the case where $\pSingle{\vTPar}$ belongs to an exponential family~\eqref{eq:exp_family}. In this case, both the assumptions and the \ac{EER} coefficient $r(\xQ; \xAll) $ can be simplified. In particular, the \ac{FIM} and the Hessian of $\log \pSingle{\vPar}$ reduce to 
\begin{align}\label{eq:property_exp_family}
\FIM(\vPar; x) = - \frac{\partial^2}{\partial {\vPar} \, \partial {\vPar\T} }  \log \pSingle{\vPar}(y | x) = \nabla^2 b(\vPar, x) .
\end{align}
Here, $\nabla^2 b(\vPar, x)$ is the Hessian matrix of $b(\cdot, x) $ viewed as a function of the first argument. 
The \ac{EER} coefficient  for the exponential family is described in the following corollary.
\begin{corollary}\label{cor:decay_coeff_exp}
Suppose $\pSingle{\vTPar}$ is given by~\eqref{eq:exp_family}. Under Assumptions~\eqref{eq:regular}--\eqref{eq:finite_fourth_moment}, equality~\eqref{eq:asymptotic_eer} holds with $r(\xQ; \xAll) $ given by
\begin{align}
r(\xQ; \xAll) & = \frac{1}{2} \mathrm{tr}\Bigl\{    \nabla^2 b(\underline{\alpha}, \xQ) \,  \Bigl( \frac{1}{\nDem} \sum_{i=1}^{\nDem} \nabla^2 b(\underline{\alpha}, x_i)   \Bigr)^{-1} \Bigr\} 
\nonumber \\
& = \frac{1}{2} \mathrm{tr}\bigl\{J(\underline{\alpha};  \xQ)  \,    \demFIM(\underline{\alpha}; \xAll)^{-1} \bigr\} 
\label{eq:coefficient_exp_family}
\end{align}
where we define 
$
\demFIM(\underline{\alpha}; \xAll) \coleq \frac{1}{\nDem}\sum_{i=1}^{\nDem}  J(\underline{\alpha}; x_i) 
$
.
\end{corollary}
\begin{IEEEproof}
Combining~\eqref{eq:property_exp_family} with \eqref{eq:def_A_n} and \eqref{eq:def_M_n}, we obtain 
\begin{align*}
A_n =  \mathbb E[M_n | \xAll]  = \frac{1}{\nDem} \sum_{i=1}^{\nDem} \nabla^2 b(\underline{\alpha}, x_i).
\end{align*}
Substituting this into~\eqref{eq:coefficient_general} gives~\eqref{eq:coefficient_exp_family}.
\end{IEEEproof}

Corollary~\ref{cor:decay_coeff_exp} shows that   $r(\xQ; \xAll) $ is determined by the \ac{FIM} $\FIM(\underline{\alpha};  \xQ)$ at the quest $\xQ$ and the average \ac{FIM} $\demFIM(\underline{\alpha}; \xAll)$ at all demonstrations. Properties of $r(\xQ; \xAll) $ are described in the following corollary, followed by their interpretation. 
\begin{corollary}\label{cor:coefficient_properties}
The \ac{EER} coefficient $r(\xQ; \xAll) $ satisfies the following properties: 
\begin{enumerate}
\item [(i)]
if another quest $\xQ'$ satisfies $\FIM(\underline{\alpha};  \xQ' ) \preceq \FIM(\underline{\alpha};  \xQ)$, then $r(\xQ'; \xAll) \leq r(\xQ; \xAll)$; 
\item [(ii)] 
under the constraint $\det \bigl( \demFIM(\underline{\alpha}; \xAll) \bigr) \leq \det\bigl(J(\underline{\alpha};  \xQ) \bigr) $, it holds that 
$ r(\xQ; \xAll) \geq {\dimPar}/{2}  $, 
with equality achieved if $\demFIM(\underline{\alpha}; \xAll) = \FIM (\underline{\alpha};  \xQ) $.
\end{enumerate}
\end{corollary}
\begin{IEEEproof}
Using the \ac{PSD} property  of \acp{FIM}  and
$
r(\xQ; \xAll) = 
\frac{1}{2}\mathrm{tr}\bigl\{ \demFIM(\underline{\alpha}; \xAll)^{-1/2} \FIM(\underline{\alpha};  x)  \, \demFIM(\underline{\alpha}; \xAll)^{-1/2}  \bigr\}
$ for both $x = \xQ$ and $x = \xQ'$, we obtain (i). To prove (ii), we use $\lambda_i(A)$ to denote the $i$-th largest eigenvalue of a symmetric matrix A. Von Neumann's trace inequality gives
\begin{align*}
r(\xQ; \xAll)
&\geq \frac{1}{2} \sum_{k=1}^{\dimPar} \lambda_k\bigl(\FIM(\underline{\alpha};  \xQ)  \bigr) \, \lambda_{\dimPar -k + 1} \bigl(\demFIM(\underline{\alpha}; \xAll)^{-1} \bigr) \\
& =  \frac{1}{2}  \sum_{k=1}^{\dimPar}  \frac{\lambda_k \bigl(\FIM(\underline{\alpha};  \xQ)  \bigr) }{  \lambda_{k} \bigl(\demFIM(\underline{\alpha}; \xAll) \bigr)} \\
& \geq \frac{\dimPar}{2} \biggl(\prod _{k=1}^{\dimPar} \frac{\lambda_k \bigl(\FIM(\underline{\alpha};  \xQ) \bigr ) }{  \lambda_{k} \bigl(\demFIM(\underline{\alpha}; \xAll) \bigr) }\biggr)^{1/\dimPar} \\
& = \frac{\dimPar}{2} \biggl(\frac{\det \bigl(\FIM(\underline{\alpha};  \xQ) \bigr)}{\det \bigl(\demFIM(\underline{\alpha}; \xAll) \bigr)} \biggr)^{1/\dimPar} \geq \frac{\dimPar}{2} .
\end{align*}
Equality condition can be verified via direct calculation.
\end{IEEEproof}

\begin{remark}
Part (i) of Corollary~\ref{cor:coefficient_properties} shows that the asymptotic \ac{EER} is positively correlated with the sensitivity of the conditional distribution $\pSingle{\vTPar}$ with respect to variations of the parameter. Specifically, the \ac{FIM} measures the sensitivity of $\pSingle{\vTPar}(\cdot | x)$ with respect to variations of the parameter at $\vTPar$ for demonstration $x$. If $\FIM(\underline{\alpha};  \xQ' ) \preceq \FIM(\underline{\alpha};  \xQ)$, then $\pSingle{\vTPar}(\cdot | \xQ')$ is less sensitive to variations of the parameter for demonstration $\xQ'$ than for $\xQ$. In other words, the deviation of $\pSingle{\vMLEPar}(\cdot | \xQ')$ from $\pSingle{\vTPar}(\cdot | \xQ')$ is less significant than the deviation of $\pSingle{\vMLEPar}(\cdot | \xQ)$ from $\pSingle{\vTPar}(\cdot | \xQ)$. Consequently, the asymptotic \ac{EER} for $\xQ'$ is smaller than that for $\xQ$. 

Part (ii) of Corollary~\ref{cor:coefficient_properties} shows that under a determinant constraint, the asymptotic \ac{EER} is minimized when the average \ac{FIM} at all demonstrations matches the \ac{FIM} at the quest.  This indicates that the \ac{ICL} performance can be improved if one carefully designs the demonstrations to bring the average \ac{FIM} at these demonstrations close to the \ac{FIM} at the quest. Note that the determinant constraint acts as a normalization condition and cannot be dropped. Otherwise the minimal asymptotic \ac{EER} would be achieved when the \ac{FIM} at each demonstration is maximized while meeting Assumption~\eqref{eq:finite_fim_dem}, i.e.,  $\FIM(\vPar; x_i) = c_1 I$ for all $i=1,2,\dotsc, n$. This is a case of little theoretical insights.
\end{remark}

Finally, we give an example for the exponential family. Consider the case where the conditional distribution of $Y \in \mathbb R^{\dimPar}$ given $X\in \mathbb R^{\dimPar}$ is Gaussian.
Specifically, $\pSingle{\vTPar}$ is given by 
\begin{align}
\pSingle{\vTPar} (y|x) &= \frac{1}{\sqrt{2\pi}^{\dimPar} \prod_{k=1}^{\dimPar} [x]_k} \nonumber \\
& \times \exp \Bigl(-\frac{1}{2} (y-\vTPar)\T \diag(x)^{-2} (y-\vTPar)\Bigr) 
\label{eq:pdf_gauss}
\end{align}
where $\diag(x) \in \mathbb R^{\dimPar \times \dimPar}$ is a diagonal matrix with the entry on its $k$-th row and $k$-th column being $[x]_k$.
Distribution~\eqref{eq:pdf_gauss} belongs to the exponential family~\eqref{eq:exp_family} with $f(x,y) = \diag(x)^{-2} y$ and $b(\vTPar, x) = \frac{1}{2} \vTPar\T \diag(x)^{-2} \vTPar + \sum_{k=1}^{\dimPar} \log [x]_k$. The \ac{FIM} $\FIM(\vTPar; x) = \diag(x)^{-2}$. Assumptions~\eqref{eq:regular}--\eqref{eq:finite_fourth_moment} hold if there exist constants $c_1 > c_2 > 0$ such that $c_1^{-1/2} \leq [x]_k \leq c_2^{-1/2}$ for $x = \xQ$ and for $x = x_i$ with $i=1,2,\dotsc, \nDem$. Moreover,  Corollary~\ref{cor:decay_coeff_exp} gives $r(\xQ; \xAll) = \frac{1}{2n} \sum_{k=1}^{\dimPar} \sum_{i=1}^{\nDem} \bigl([x_i]_k^2 / [\xQ]_k^2 \bigr)$. 
 
\section{Non-Asymptotic Bound of \ac{EER}}
The next theorem provides an upper bound on $\eer(\xQ; x^n) $ for the exponential family without assuming that $\nDem$ is large.
\begin{theorem}
Suppose $\pSingle{\vTPar}$ belongs to an exponential family~\eqref{eq:exp_family}. Under Assumptions~\eqref{eq:regular}--\eqref{eq:finite_fourth_moment}, it holds that
\begin{align}\label{eq:non_asymp_ub}
\ell(x_Q;x^n) \leq & r(x^n) \frac{1}{n} + (a_1 + a_2) \frac{1}{n^{1.5}} + a_3 \frac{1}{n^{2}} 
\end{align}
where $a_1$, $a_2$, and $a_3$ are defined as
\begin{align}
a_1 &:= \frac{1}{2c_2^2} c_1^{3/2} c_3 K^2 \Bigl((K^2 + 2K) c_1^2 + \frac{c_4}{n}\Bigr)^{1/2} \\
a_2 & :=\frac{1}{6} c_3 K^{3/2}\Bigl( (K^2 + 2K) c_1^2 + \frac{c_4}{n}\Bigr)^{3/4} \\
a_3 & := \frac{1}{8} \frac{c_1}{c_2^2} c_3^2 K^3  \Bigl( (K^2 + 2K) c_1^2 + \frac{c_4}{n}\Bigr) .
\end{align}
\end{theorem}
\begin{IEEEproof}
The key to the proof is to perform higher-order Taylor expansion of $\KLD \bigl(  \pSingle{\vTPar} \, \big \| \, \pSingle{\vMLEPar}   \bigr) $ and $l(\vMLEPar; Y^n)$ than in~\eqref{eq:kld_taylor_general} and~\eqref{eq:relation_alpha}. 
In particular, applying third-order Taylor expansion of $g(\vMLEPar)$ at $ \underline{\alpha}$ and using~\eqref{eq:g_gradient_zero}, we obtain 
\begin{align}\label{eq:kld_taylor_third_order}
\KLD \bigl(  \pSingle{\vTPar} \, \big \| \, \pSingle{\vMLEPar}   \bigr) 
= -\frac{1}{2} (\underline{\alpha}^* - \underline{\alpha})^{\mathrm T} \, \nabla^2 g(\vTPar) \, (\underline{\alpha}^* - \underline{\alpha}) + e_0
\end{align} 
where $e_0$ is the third-order term in the Tylor expansion:
\begin{align*}
e_0 \coleq \sum_{ v: \| v\|_1 = 3} \frac{1}{ v !} \frac{\partial^3}{\partial \xi^{ v}} b(\xi, x_Q) \, (\xi - \underline{\alpha})^{ v} .
\end{align*}
Here, $ v =[v_1 \quad v_2 ~ \cdots ~v_K]^{\mathrm T}$ is a vector consisting of $\dimPar$ non-negative inters, $ v! \coleq  \prod_{k=1}^{\dimPar} (v_k!)$, vector $\xi$ is on the line segment connecting $\vTPar$ and $\vMLEPar$. Moreover,  
$\frac{\partial^3}{\partial \xi^{ v}}  := \frac{\partial^3}{\partial \xi_1^{v_1} \partial \xi_2^{v_2} \cdots \partial \xi_K^{v_K}}$, and
$
y^{ v} := \prod_{i=1}^K y_k^{v_k} $ for  $y \in \mathbb R^K $.
In addition, applying second-order Taylor expansion of $l(\vMLEPar; Y^n)$ at $\vTPar$ and using the property of \ac{MLE} that  $l(\vMLEPar; Y^n) = 0$ , we obtain
\begin{align} \label{eq:l_expansion_second_order}
- l(\underline{\alpha}; Y^n) = \Jac l(\vTPar ; Y^n) \, (\vMLEPar - \vTPar  ) + u
\end{align}
where $u \in \mathbb R^{\dimPar}$ with its $k$-th entry defined as
$ \frac{1}{2} (\vMLEPar - \vTPar  )\T \nabla^2 l_k(\vTPar; \ryAll) (\vMLEPar - \vTPar)$.  Here, $l_k(\vTPar; \ryAll)$ represents the $k$-th component of $l(\vTPar; \ryAll)$. Equation~\eqref{eq:l_expansion_second_order} gives $\vMLEPar - \vTPar = -\bigl( \Jac l(\vTPar ; Y^n) \bigr)^{-1} (l(\underline{\alpha}; Y^n) + u) $. Substituting this into~\eqref{eq:kld_taylor_third_order} and taking conditional expectation over $\pMulti{\vTPar}$ given $\rxAll = \xAll$, we obtain $\eer(\xQ; x^n) = e_1 + e_2 + e_3 + e_0 $,  where
\begin{align*}
e_1 & = \frac{1}{2} \mathbb E \bigl[l(\underline{\alpha}; Y^n) \T \, G_n \, l(\underline{\alpha}; Y^n) \bigr | \xAll \bigr ]  \nonumber \\
e_2  & =   \mathbb E \bigl[u \T G_n l(\underline{\alpha}; Y^n) \bigr | \xAll \bigr ], \qquad e_3 = \frac{1}{2} \mathbb E \bigl [u \T G_n  u  \bigr | \xAll \bigr ] 
\end{align*}
with $G_n \coleq  \bigl( \Jac l(\vTPar ; Y^n) \bigr)^{-1} \nabla^2 b(\vTPar, x_Q) \bigl( \Jac l(\vTPar ; Y^n) \bigr)^{-1}$.  
Triangle inequality gives $\eer(\xQ; x^n) \leq  | e_1 | + | e_2|  + | e_3|  + | e_0| $. In particular, $|e_1|$ can be shown to be ${r(\xQ; \xAll)}/{n}$. Moreover, $|e_2|$, $|e_3|$, and $|e_4|$ can be upper bounded by $a_1/n^{1.5}$, $a_3/n^2$, and $a_2/n^{1.5}$, respectively. Combining these gives~\eqref{eq:non_asymp_ub}. 
\end{IEEEproof}

\section{conclusion}
In this work, we propose a probabilistic model for \ac{ICL} and analyze its  \ac{EER}. For general parametric distributions, we show that the asymptotic \ac{EER} is the ratio between the \ac{EER} coefficient and the number of demonstrations. For exponential families, the \ac{EER} coefficient is a function of the \ac{FIM} at the quest and of the \ac{FIM} averaged over all the demonstrations. Moreover, a non-asymptotic upper bound on the \ac{EER} is derived for exponential families. Based on these results, we show the effect on \ac{ICL} performance of different factors including the number of demonstrations, the sensitivity of the probabilistic model to the variation of its parameters, as well as the similarity between the demonstrations and the query. This work sheds theoretical insights on \ac{ICL} and  provides design guidelines for model training and prompt engineering in \ac{LLM}.

\appendix[Proof of Theorem~\ref{prop:asymptotics_general}]
\begin{IEEEproof}
The proof idea is described as follows. Viewing $\KLD \bigl(  \pSingle{\vTPar} \, \big \| \, \pSingle{\vMLEPar}   \bigr) $ as a function of $\vMLEPar$, we use Taylor expansion to transform this function to a quadratic form of $\vMLEPar - \vTPar$. Then, we show the $L^2$ convergence of  $\| \vMLEPar - \vTPar \| $  as $\nDem \to \infty$. Finally,  we use such a convergence to construct an approximation $\tilde D(\ryAll, \xAll)$ of $\KLD \bigl(\pSingle{\vTPar}   \, \bigr \| \, \pSingle{\vMLEPar} \bigr) $ such that 
\begin{subequations}\label{eq:tilde_kld_properties}
\begin{align}
\! \! \lim_{n\to\infty} n \mathbb E \bigl[\KLD \bigl(\pSingle{\vTPar}   \, \bigr \| \, \pSingle{\vMLEPar} \bigr)  - \tilde D(\ryAll, \xAll) \bigr | x^n \bigr] &= 0 \label{eq:tilde_kld_approximation} \\
\!  \lim_{n\to\infty} n \mathbb E \bigl[  \tilde D(\ryAll, \xAll)  \big | x^n \bigr] - r(\xQ; x^n) &= 0 .
  \label{eq:tilde_kld_expectation}
\end{align}
\end{subequations}

Details of the proof are provided in the following. Define function $g : \mathbb R^{\dimPar} \mapsto \mathbb R$ as
\begin{align}\label{eq:def_g}
g(\alpha) \coleq \mathbb E_{\pSingle{\vTPar}} \bigl [\log  \pSingle{\vPar} (Y| \xQ)  \bigr | X = \xQ \bigr] .
\end{align}
Note that $\vPar$ does not affect the distribution in the subscript of the expectation operator. 
Using this definition, the \ac{KL}-divergence can be written as
\begin{align}
\KLD \bigl(  \pSingle{\vTPar} \,\| \, \pSingle{\vMLEPar}   \bigr) 
&= - \bigl(g(\underline{\alpha}^*) - g(\underline{\alpha}) \bigr) .
\end{align}
Using Assumption~\eqref{eq:regular}, we obtain 
\begin{align} \label{eq:g_gradient_zero}
\nabla g(\underline{\alpha}) = 
 \mathbb E_{\pSingle{\vTPar}}\Bigl[\frac{\partial }{\partial {\underline{\alpha}}} \log  \pSingle{\vTPar}(Y| \xQ ) \Bigr | X = \xQ \Bigr]
= 0 .
\end{align}
Applying second-order Taylor expansion of $g(\vMLEPar)$  at $\underline{\alpha}$ and using~\eqref{eq:g_gradient_zero}, we obtain
\begin{align} \label{eq:kld_taylor_general}
\KLD \bigl(  \pSingle{\vTPar} \, \big \| \, \pSingle{\vMLEPar}   \bigr) 
= -\frac{1}{2} (\underline{\alpha}^* - \underline{\alpha})^{\mathrm T} \, \nabla^2 g(\xi) \, (\underline{\alpha}^* - \underline{\alpha})
\end{align}
where $\xi \in \mathbb R^{\dimPar}$ is on the line segment connecting $\underline{\alpha}$ and $\underline{\alpha}^*$.

Next, we consider the property of $\underline{\alpha}^* - \underline{\alpha}$. To that end, define the average score function $l(\alpha; Y^n)$ as 
\begin{align}\label{eq:def_log_likelihood}
l(\alpha; Y^n) := \frac{1}{n} \sum_{i=1}^n \frac{\partial}{\partial \alpha} \log \pSingle{\vPar} (Y_i | x_i ) .
\end{align}
The definition of \ac{MLE} gives $l(\vMLEPar; Y^n) = 0$. 
Combining this with the mean-value theorem, we obtain
\begin{align}\label{eq:relation_alpha}
l(\underline{\alpha}; Y^n) = \Jac l(\eta ; Y^n) \, (\underline{\alpha} - \underline{\alpha}^*) .
\end{align}
where $\eta \in \mathbb R^{\dimPar}$ is on the line segment connecting $\underline{\alpha}$ and $\underline{\alpha}^*$, and $ \Jac l(\eta ; Y^n) \coleq \frac{\partial l(\eta ; Y^n)}{ \partial \eta \T}$ represents the Jacobian matrix of $ l(\eta ; Y^n)$. Expressions~\eqref{eq:def_log_likelihood} and~\eqref{eq:positive_fim} show that $ \Jac l(\eta ; Y^n)$ is invertible. Consequently,  
\begin{align}
\underline{\alpha} - \underline{\alpha}^* = 
\bigl( \Jac l(\eta ; Y^n) \bigr)^{-1} \,l(\underline{\alpha}; Y^n) .
\label{eq:alpha_taylor}
\end{align}
%
Using~\eqref{eq:alpha_taylor} and~\eqref{eq:positive_fim}, we have
\begin{align}
\mathbb E \bigl[ \| \vTPar - \vMLEPar\|^2 \bigr | \xAll \bigr]
& \leq \mathbb E \bigl[ \bigl\|  \bigl( \Jac l(\eta ; Y^n) \bigr)^{-1} \bigr\|^2 \, \bigl\| l(\underline{\alpha}; Y^n) \bigr\|^2  \bigr | \xAll \bigr] \nonumber \\
& \leq \frac{1}{c_2^2}  \mathbb E \bigl[   \bigl\| l(\underline{\alpha}; Y^n) \bigr\|^2  \bigr | \xAll \bigr] \nonumber \\
& =  \frac{1}{c_2^2 \nDem^2} \sum_{i=1}^n \tr\{\FIM(\vTPar; x_i) \} \leq \frac{c_1 \dimPar}{c_2^2 \nDem}
\label{eq:MLE_approaches_true}
\end{align}
where \eqref{eq:finite_fim_dem} is used in the last inequality. 

Finally, we construct $\tilde D(\ryAll, \xAll)$ and show~\eqref{eq:tilde_kld_properties}. Substituting~\eqref{eq:alpha_taylor} into~\eqref{eq:kld_taylor_general} and using a property of trace gives 
\begin{align}
 D_{\mathrm{KL}}\bigl(\pSingle{\vTPar}   \,\| \, \pSingle{\vMLEPar} \bigr)   = \frac{1}{2n} \mathrm{tr}\bigl\{ &  M_n \, \bigl(\Jac l(\eta ; y^n) \bigr)^{-1} \nabla^2 g(\xi)  \nonumber \\
&
\times     \bigl(\Jac l(\eta ; y^n) \bigr)^{-1} \bigr\} .
  \label{eq:kld_taylor_trace}
\end{align}
Define $\tilde D(\ryAll, \xAll)$  by replacing $\eta$ and $\xi$ in~\eqref{eq:kld_taylor_trace} with $\underline{\alpha}$, i.e., 
\begin{align}
 \tilde D(\ryAll, \xAll)   \coleq
\frac{1}{2n} \mathrm{tr}\bigl\{ & M_n \, \bigl(\Jac l(\underline{\alpha} ; y^n) \bigr)^{-1} \nonumber \\
& \times \nabla^2 g(\underline{\alpha}) \,  \bigl(\Jac  l(\underline{\alpha}; y^n) \bigr)^{-1} \bigr\}
\label{def:tilde_kld}
\end{align}
We can show $\tilde D(\ryAll, \xAll)$ satisfies~\eqref{eq:tilde_kld_approximation}. To that end,  define 
\begin{align}
L_n  & \coleq  \bigl(\Jac l(\eta ; y^n) \bigr)^{-1} \nabla^2 g(\xi)  \bigl(\Jac l(\eta ; y^n) \bigr)^{-1} \nonumber \\
&  \qquad -  \bigl(\Jac l(\underline{\alpha} ; y^n) \bigr)^{-1}   \nabla^2 g(\underline{\alpha}) \,  \bigl(\Jac  l(\underline{\alpha}; y^n) \bigr)^{-1} .
\label{eq:def_L_n}
\end{align}
Combining this definition with~\eqref{eq:kld_taylor_trace} and~\eqref{def:tilde_kld} gives
\begin{align}
& n \mathbb E \bigl [D_{\mathrm{KL}}\bigl(\pSingle{\vTPar}   \,\| \, \pSingle{\vMLEPar} \bigr)  -   \tilde D(\ryAll, \xAll) \bigr |  \xAll  \bigr ] \nonumber \\
& \quad = \frac{1}{2} \mathbb E \bigl[\mathrm{tr} \{M_n L_n\} \big| x^n \bigr]   \nonumber \\
& \quad \leq \frac{1}{2} \mathbb E\bigl[\| M_n\|_{\mathrm F} \, \| L_n\|_{\mathrm F}  \big | x^n \bigr] \nonumber \\
& \quad \leq \frac{1}{2} \mathbb E\bigl[ \| M_n\|_{\mathrm F}^2 \big | x^n\bigr] \, \mathbb E\bigl[ \| L_n\|_{\mathrm F}^2 \big | x^n\bigr]
\end{align}
where Cauchy—Schwartz inequalities are used in the last two inequalities. Using~\eqref{eq:finite_fim_dem} and~\eqref{eq:finite_fourth_moment}, we can show that $\sup_n \mathbb E\bigl[ \| M_n\|_{\mathrm F}^2 \big | x^n\bigr] $ is finite. Moreover combining~\eqref{eq:MLE_approaches_true} with~\eqref{eq:finite_fim_quest},~\eqref{eq:positive_fim}, and~\eqref{eq:finite_third_order_dem}, we can show that $\lim_{n \to \infty} \mathbb E\bigl[ \| L_n\|_{\mathrm F}^2 \big | x^n\bigr] = 0$. Combining this with $\sup_n \mathbb E\bigl[ \| M_n\|_{\mathrm F}^2 \big | x^n\bigr] < \infty $ gives~\eqref{eq:tilde_kld_approximation}.  Moreover, combining~\eqref{eq:def_g} with \eqref{eq:def_fim}, and combining~\eqref{eq:def_log_likelihood} with~\eqref{eq:def_A_n}, we obtain
$
-\nabla^2 g(\underline{\alpha}) = \FIM(\vTPar; \xQ)
$
and 
$\Jac l(\eta ; y^n) = - A_n $, 
respectively. 
Using these two expressions gives~\eqref{eq:tilde_kld_expectation}. Using~\eqref{eq:tilde_kld_properties}, we obtain the desired result.
\end{IEEEproof}

\bibliographystyle{IEEEtran}
\bibliography{IEEEabrv, temp}

\end{document}